\documentclass[runningheads]{llncs}

\usepackage[disable]{todonotes}
\usepackage{enumerate}
\usepackage{amsmath}
\usepackage{amssymb}
\usepackage{amsfonts}
\usepackage{latexsym}
\usepackage{xspace}
\usepackage{thm-restate}
\usepackage{graphicx}
\usepackage[misc,geometry]{ifsym}
\usepackage{subcaption}
\usepackage[inline]{enumitem}

\definecolor{defblue}{rgb}{0.121,0.47,0.705}
\definecolor{linkblue}{rgb}{0.098,0.098,0.4392}
\let\emph\relax
\DeclareTextFontCommand{\emph}{\color{defblue}\em}

\usepackage[colorlinks=true,
			linkcolor = linkblue,
			anchorcolor = linkblue,
			citecolor = linkblue,
			filecolor = linkblue,
			menucolor = linkblue,
			urlcolor = linkblue,
			bookmarks=true,
			bookmarksopen=true,
			bookmarksopenlevel=2,
			bookmarksnumbered=true,
			hyperindex=true,
			plainpages=false,
			pdfpagelabels=true,
]{hyperref} 
\usepackage[capitalise,noabbrev,nameinlink]{cleveref}

\newcommand\unobstructed{unobstructed\xspace}
\graphicspath{{figs/}}

\newcommand{\NP}{NP}

\renewcommand{\orcidID}[1]{\href{https://orcid.org/#1}{\includegraphics[scale=.03]{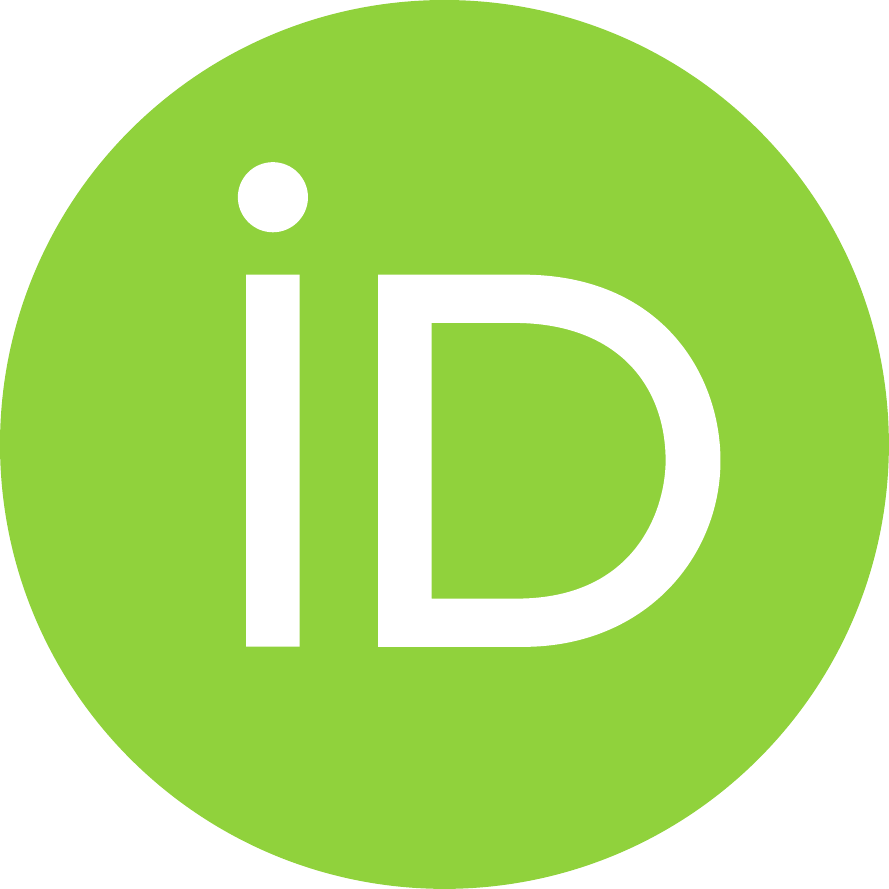}}}

\usepackage[ruled,vlined]{algorithm2e}

\DontPrintSemicolon

\usepackage{tabularx}

\title{One-Bend Drawings of Outerplanar Graphs Inside Simple Polygons%
\thanks{This work was initiated at the Workshop on Graph and Network Visualization 2019. We thank all the participants for helpful discussions and Anna Lubiw for bringing the problem to our attention.}}

\author{
Patrizio Angelini\inst{1}\orcidID{0000-0002-7602-1524}
	\and
Philipp Kindermann\inst{2}\orcidID{0000-0001-5764-7719}
	\and
Andre L\"offler\inst{3}%
	\and
Lena Schlipf\inst{4}\thanks{This research is supported by the Ministry of Science, Research and the Arts Baden-W\"urttemberg (Germany).}\orcidID{0000-0001-7043-1867}
	\and
Antonios Symvonis\inst{5}\orcidID{0000-0002-0280-741X}
}
\institute{
John Cabot University, Rome, Italy\\
\texttt{pangelini@johncabot.edu}
	\and
Universit\"at Trier, Trier, Germany\\
\texttt{kindermann@uni-trier.de}
	\and
Universit\"at W\"urzburg, W\"urzburg, Germany\\
\texttt{andre.loeffler@uni-wuerzburg.de}
	\and
Universit\"at T\"ubingen, T\"ubingen, Germany\\
\texttt{schlipf@informatik.uni-tuebingen.de}
	\and
National Technical University of Athens, Athens, Greece\\
\texttt{symvonis@math.ntua.gr}
}

\authorrunning{P. Angelini, P. Kindermann, A. L\"offler, L. Schlipf, and A. Symvonis}

\begin{document}

\maketitle

\begin{abstract}
  We consider the problem of drawing an outerplanar graph with $n$ vertices 
  with at most one bend per edge if the outer face is already drawn as a simple polygon.
  We prove that it can be decided in $O(nm)$ time if such a drawing
  exists, where $m\le n-3$ is the number of interior edges. In the positive case, we can also compute such a drawing.
	
\end{abstract}

\section{Introduction}

One of the fundamental problems in graph drawing is to draw a planar graph crossing-free under certain geometric or topological constraints. 
Many classical algorithms draw planar graphs under the constraint that all edges have to be straight-line segments~\cite{fraysseix1990draw,schnyder90,tutte1963draw}. %
In practical applications, however, we do not always have the freedom of drawing the whole graph from scratch, as some important parts of the graph may already be drawn.

This problem is known as the \textsc{Partial Drawing Extensibility} problem.
Formally, given a planar graph $G=(V,E)$, a subgraph $H=(V',E')$ with $V'\subseteq V$ and $E'\subsetneq E$ called the \emph{host} graph, and a planar drawing $\Gamma_H$ of host~$H$, the problem asks for a planar drawing $\Gamma_G$ of~$G$ such that the drawing of $H$ in $\Gamma_G$ coincides with $\Gamma_H$. 
This problem was first proposed by Brandenburg et al.~\cite{brandenburg03} in 2003
and has received a lot of attention in the subsequent years.

\paragraph{Problem Statement.}

In this paper, we consider a special drawing extension setting, in which~$G$ 
is a biconnected outerplanar graph with $n$ vertices and $m$ interior edges,
the host graph~$H$ is the simple cycle bounding the outer face of $G$,
and $\Gamma_H$ is a 1-bend drawing of $H$.
In other words, we have  a simple polygon~$P=\Gamma_H$ as input and we want to draw the interior 
edges of $G$ inside $P$ without crossings.
Testing for a straight-line extension is trivial.
Moreover, for any integer $k$, there exists some instances that have a 
$k$-bend extension but no $(k-1)$-bend extension; see, e.g., Fig.~\ref{fig:two-bends} for $k=2$.
Hence, it is of interest to test for a given~$k$ whether a $k$-bend extension exists. 
In this paper, we present an algorithm that solves this problem for $k=1$ in $O(nm)$ time. 
More formally, we prove the following theorem.

\begin{restatable}{theorem}{MainTheorem}\label{thm:main}
Given a biconnected outerplanar graph $G=(V,E)$ with $n$ vertices and $m$ interior edges,
and a 1-bend drawing of the outer face of $G$,
we can decide in $O(nm)$ time whether $G$ admits an outerplanar drawing with at most one bend per edge
that contains the given drawing of the outer face.
In the positive case, we can also compute such a drawing.
\end{restatable}

\paragraph{Related Work.}

For the case of extending a given straight-line drawing of a planar graph using straight-line segments as edges, Patrignani~\cite{patrignani06} showed the problem to be \NP{}-hard, but he could not prove membership in \NP{}, as a solution may require coordinates not representable with a polynomial number of bits.
Recently, Lubiw, Miltzow, and Mondal~\cite{lubiw18} proved that a generalization of the problem, where overlaps (but not proper crossings) between edges of $E\setminus E'$ and $E'$ are allowed, is hard for the existential theory of the reals $(\exists\mathbb{R}$-complete).

These results motivate allowing bends in the drawing.
Angelini et al.~\cite{angelini15} presented an algorithm to test in linear time whether there exists any topological planar drawing of~$G$ with pre-drawn subgraph, and Jel{\'{\i}}nek, Kratochv{\'{\i}}l, and Rutter~\cite{jelinek13} gave a characterization via forbidden substructures. 
Chan et al.~\cite{chan15} showed that a linear number of bends ($72|V'|$) per edge suffices.
This number is also asymptotically worst-case optimal as shown by Pach and Wenger~\cite{pach01} for the special case of the host graph not containing edges ($E'=\emptyset$).

Special attention has been given to the case that the host graph $H$ is exactly the outer face of~$G$. 
Already Tutte's seminal paper~\cite{tutte1963draw} showed how to obtain a straight-line convex drawing of a triconnected planar graph with its outer face drawn as a prescribed convex polygon. 
This result has been extended by Hong and Nagamochi~\cite{hong08} to the case that the outer face is drawn as a star-shaped polygon without chords.
Mchedlidze and Urhausen~\cite{mchedlidze18} study the number of bends required based on the shape of the drawing of the outer face and show that one bend per edge suffices if the outer face is drawn as a star-shaped polygon.
Mchedlidze, N\"ollenburg, and Rutter~\cite{DBLP:journals/algorithmica/MchedlidzeNR16}
give a linear-time algorithm to test for the existence of a straight-line drawing of~$G$ in the case that $H$ is biconnected and $\Gamma_H$ is a convex drawing.
Ophelders et al.~\cite{DBLP:conf/compgeom/OpheldersRSV21} characterize
the instances in which $G$ is a (planar) graph and $H$ is a cycle that admit
a positive solution for {\em any} straight-line drawing $\Gamma_H$ as the outer face.

\begin{figure}[t]
	\centering
	
\begin{subfigure}[t]{.28\textwidth}
	\centering
	\includegraphics{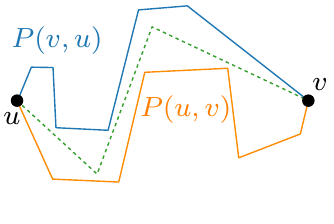}
	\caption{}
	\label{fig:two-bends}
\end{subfigure}
\hfill
\begin{subfigure}[t]{.34\textwidth}
	\centering
	\includegraphics[page=1]{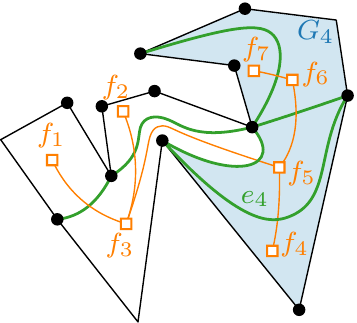}
	\caption{}
	\label{fig:dual-tree}
\end{subfigure}
\hfill
\begin{subfigure}[t]{.34\textwidth}
	\centering
	\includegraphics[page=3]{figs/dual-tree}
	\caption{}
	\label{fig:subproblem}
\end{subfigure}
 \caption{
 (a)~The edge $e=(u,v)$ requires two bends.
(b)~A biconnected outerplanar graph and its dual tree rooted in $f_7$.
(c)~A drawing of $G_4$ in $P_4$.} 
 \label{fig:dual-double-reflex}
\end{figure}

\paragraph{Notation and Preliminaries.}
\label{sec:notation}

For a pair of vertices $u$ and $v$, we denote by $\overline{uv}$ the straight-line segment between them. 
Starting at $u$ and following the boundary $\partial P$ of $P$ in counterclockwise order until reaching $v$, we obtain the interval $P(u,v)$.

The faces of $G$ induce a unique dual tree $T$~\cite{Proskurowski86} where each edge of $T$ corresponds to an \emph{interior} edge of $G$; see Fig.~\ref{fig:dual-tree}.
In the following, we consider $T$ to be rooted at some degree-1 node $f_{m+1}$.

We will traverse the dual tree twice -- first bottom-up and then top-down.
In the bottom-up traversal, we incrementally process the interior edges of $G$ 
and refine $P$ by cutting off parts that cannot be used anymore.
In the top-down traversal, we compute positions for the bend points inside the refined polygons.

We label the faces of $G$ as $f_1,\ldots,f_{m+1}$ according to the bottom-up traversal.
This sequence also implies a sequence of subtrees of $T$.
For step $i = 1, \dots, m$, let $T_i$ be the subtree of $T$ induced by the nodes $f_i,\dots,f_{m+1}$. 
We denote by $p(f_i)$ the parent of $f_i$ in $T$, and by $e_i$ the edge of $G$ corresponding to the edge of $T$ between $f_i$ and its parent.
We say that $f_i$ and $f_j$ are \emph{siblings} if $p(f_i)=p(f_j)$. 
We denote by $G_i$ the subgraph of $G$ induced by the vertices incident to the faces $f_i,\ldots f_{m+1}$, hence $G=G_1$.
Similar to the sequence of subtrees, we also have a corresponding sequence of refined polygons $P=P_1\subseteq\ldots\subseteq P_{m+1}$, where $P_{i}$ contains (at least) the vertices of~$G_{i}$. %

We draw the edge $e_i$ inside $P_i$ with one bend point $b$. This drawing of $e_i$ 
cuts~$P_i$ into two parts; the part that contains all the vertices of~$G_{i+1}$ is the \emph{\unobstructed region} $U_{i}(b)$ of $b$, the other part is the \emph{obstructed region} $O_{i}(b)$ of $b$; see Fig.~\ref{fig:Oi}.
We classify edge $e_i$ based on the type of corner that $b$ will induce in $U_{i}(b)$.
Edge $e_i$ is either
\begin{enumerate*}[label=(\roman*)]
 \item a \emph{convex} edge if it is possible to draw $e_i$ such that $b$ is a convex corner in $U_{i}(b)$, or
 \item a \emph{reflex} edge otherwise, i.e., any drawing of $e_i$ results in a strictly reflex corner at~$b$ in $U_i(b)$; see Fig.~\ref{fig:reflex}. 
\end{enumerate*}

\section{Refining the Polygon}
\label{sec:procedure}

During the bottom-up traversal of $T$, we maintain the invariant that 
$G_{i+1}$ can be drawn in $P_{i+1}$ if and only if $G_i$ can be drawn in $P_i$. 
Observe that, while $P_1=P$ represents a 1-bend drawing of the outer face of $G_1=G$,
in general some edges on the outer face of $G_i$ might be drawn in the interior
of $P_i$; see Fig.~\ref{fig:subproblem}.
To this end, to obtain $P_{i+1}$, we want to refine $P_i$ in the least restrictive way -- 
cutting away as little of $P_i$ as possible.
In particular, we will choose $P_{i+1}$ as the union of the \unobstructed regions for all possible bend points of $e_i$.

Among all leaves of $T_i$, we choose the next node $f_i$ to process as follows:
If~$T_i$ has a leaf corresponding to a reflex edge, then we process the 
corresponding interior edge next.
Otherwise, all leaves in $T_i$ correspond to convex edges, and we choose one 
of the nodes of the dual tree among them that has the largest distance to the 
root $f_{m+1}$ in $T$. 
We do this to make sure that a convex edge is only chosen if all siblings 
corresponding to reflex edges have already been processed.
The idea is that for a reflex edge we can determine its ``best'' drawing, 
namely one that cuts off only a part of the polygon that would be cut off by 
any valid drawing of this edge (see Lemma \ref{lem:reflex}), while for convex 
edges this is generally not possible. 
Furthermore, drawing a reflex edge may restrict the possible drawings for
siblings corresponding to convex edges, while drawing convex edges (with a 
convex bend point) does not (see Lemma~\ref{lem:convex-interaction}), so
after processing all siblings corresponding to reflex edges, we can compute
all possible bend points for a convex edge.

\begin{figure}[t]
	\centering
	\begin{subfigure}[t]{.45\textwidth}
	\centering
	\includegraphics[page=2]{figs/convex}
	\caption{}
	\label{fig:Oi}
\end{subfigure}
\hfill
\begin{subfigure}[t]{.45\textwidth}
	\centering
	\includegraphics[page=2]{figs/reflex}
	\caption{}
	\label{fig:reflex}
\end{subfigure}

\medskip

\begin{subfigure}[t]{.45\textwidth}
	\centering
	\includegraphics[page=1]{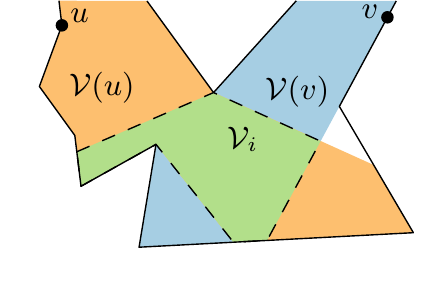}
	\caption{}
	\label{fig:Vi}
\end{subfigure}
\hfill
\begin{subfigure}[t]{.45\textwidth}
	\centering
	\includegraphics[page=3]{figs/convex}
	\caption{}
	\label{fig:minbend}
\end{subfigure}
\caption{
(a)~the region $O_i(b)$ for some $b\in\mathcal{V}_i$; %
(b)~$e_i$ is reflex %
and~$b'$ is the minimal bend point.
(c--d)~$e_i=(u,v)$ is convex: 
(c)~construction of $\mathcal{V}_i$; 
(d)~construction of $p_u,p_v$, the obstructed region $O^*_i$, and the set of minimal bend points.}
\label{fig:convex}
\end{figure}

Consider a leaf $f_i$ and the corresponding interior edge $e_i=(u,v)$.
W.l.o.g. assume that $P_i(u,v)$ does not contain an edge of the root $f_{m+1}$.
Let $\mathcal{V}(u)$ and $\mathcal{V}(v)$ be the (closed) regions inside 
$P_i$ visible from $u$ and $v$, respectively,
and let $\mathcal{V}_{i} = \mathcal{V}(u) \cap \mathcal{V}(v)$ be their intersection -- the region visible by both end points of $e_i$; see Fig.~\ref{fig:Vi}.
Clearly, any valid bend point for $e_i$ needs to be inside $\mathcal{V}_{i}$;
thus, if the interior of $\mathcal{V}_{i}$ is empty, then we reject the instance.
We define a set of \emph{minimal bend points} for $e_i$:
a point $b\in \mathcal{V}_{i}$ is minimal for $e_i$ if there is no other 
point~$b'\in \mathcal{V}_{i}$ with $O_{i}(b')\subsetneq O_{i}(b)$.
Note that all minimal bend points for $e_i$ lie on $\partial \mathcal{V}_i$,
so they are not valid bend points.

For reflex edges, we have the following lemma regarding minimal bend points.

\begin{restatable}{lemma}{lemReflex}\label{lem:reflex}
  Let $e_i = (u,v)$ be a reflex edge. 
  If there is a valid drawing of $e_i$ in $P_i$, then there is a unique minimal bend point $b$ for $e_i$.
\end{restatable}
\begin{proof}
For a contradiction, let $b_1$ and $b_2$ be two minimal bend points
for $e_i$.
Since~$e_i$ is a reflex edge, both $b_1$ and~$b_2$ must lie 
on the same side of the line $uv$.
Hence, there must be a crossing $b'$ between one of the segments $\overline{ub_1},\overline{b_1v}$ and one of the segments $\overline{ub_2},\overline{b_2v}$.
Further, $O_i(b')=O_i(b_1)\cap O_i(b_2)$, so $O_i(b')\subseteq O_i(b_1)$
and $O_i(b')\subseteq O_i(b_2)$. Since $O_i(b_1)\neq O_i(b_2)$, one of $b_1$ and $b_2$ is not minimal.
\end{proof}

Based on Lemma~\ref{lem:reflex}, we construct $P_{i+1}$ as the polygon
$P_i(v,u)\circ\overline{ub}\circ\overline{bv}$ for a reflex edge $e_i$ having minimal bend point $b$. Note that we can efficiently find~$b$ since it is 
a corner of~$\mathcal{V}_i$.

For a convex edge $e_i=(u,v)$, we can no longer rely on having a single minimal bend point.
Hence, we need to refine our notation.

We define the \emph{obstructed region} ${O^*_{i}=\bigcap_{b\in \mathcal{V}_{i}}O_{i}(b)}$ of $e_i$ to be the region of $P_i$ obstructed by all valid drawings of~$e_i$ -- wherever we place the bend point of~$e_i$, all the points of $O^*_{i}$ will be cut off.
Conversely, for every point $p\in P_i\setminus O^*_{i}$, there is a placement of the bend point of $e_i$ such that $p$ is not cut off; see Fig.~\ref{fig:minbend}. In view of this, we are going to set $P_{i+1}=P_{i}\setminus O_i^*$. However, we cannot directly compute $O_i^*$ according to its definition, as this would require considering all possible bend points. Thus, in the following we describe an efficient way to compute $P_{i+1}$.

If $u$ and $v$ lie in $\mathcal{V}_i$, then let $p_u=u$ and $p_v=v$. Otherwise, let $(u,u')$ and $(v',v)$ be the segments of $P_i(u,v)$ incident to $u$ and $v$, respectively.
Shoot a ray from $u$ through $u'$.
Rotate this ray in counterclockwise direction until it hits $\mathcal{V}_{i}$;
let this point be $p_u$.
Do the same with the ray from $v$ through $v'$, rotating it clockwise; 
let the point where it hits $\mathcal{V}_{i}$ be $p_v$. Let $\mathcal{V}_i(p_u,p_v)$ be the path from $p_u$ to $p_v$ along $\partial\mathcal{V}_i$ in counterclockwise order.
Then, $P_{i+1}$ is the polygon $P_i(v,u)\circ\overline {u p_u}\circ \mathcal{V}_{i}(p_u,p_v)\circ \overline{p_v v}$.

Note that all vertices of $G_{i+1}$ lie on $P_i(v,u)$, so they are contained in $P_{i+1}$. The bend point of $e_i$ has to lie in $\mathcal{V}_i$, which is completely contained in $P_{i+1}$. Hence, no bend point of another edge $e_j, j>i$ can lie in the region $O_i^*=P_{i}\setminus P_{i+1}$, because then $e_i$ and $e_j$ would cross.

If we manage to construct a non-degenerate polygon $P_{m+1}$, i.e., its interior
is non-empty, then we declare the instance as positive and we compute a drawing of $G$ in $P$ as follows.
We first draw $G_{m+1}$ in $P_{m+1}$, by picking an arbitrary bend point $b$ from the
interior of $\mathcal{V}_m$ in $P_{m+1}$ for $e_m$.
Suppose that we have a drawing $\Gamma_{i+1}$ for $G_{i+1}$ in $P_{i+1}$
where edge $e_{i+1}=(u,v)$ is drawn with bend point $b$.
To draw $G_i$ in $P_i$, we have to find valid bend points for the edges corresponding
to the children of $f_i$
inside the polygon $P_i'=P_i(u,v)\circ\overline{vb}\circ\overline{bu}$.
Consider such an edge $e_j=(u',v')$.
If $e_j$ is reflex, then we place its bend point inside $\mathcal{V}_j$,
very close to the minimum bend point.
If $e_j$ is convex, then we place its bend point on an arbitrary point 
in the interior of $\mathcal{V}_j$ that induces a convex corner.

\section{Correctness}
\label{sec:correctness}

In this section, we show that $G_{i+1}$ can be drawn in $P_{i+1}$ if and only if $G_i$ can be drawn in~$P_{i}$. 
If $e_i$ is a reflex edge, then $P_{i+1}$ contains exactly the points not cut off
by placing $e_i$ at its unique minimal bend point $b$, as Lemma~\ref{lem:reflex} indicates. Hence, $G_{i+1}$ can be drawn in $P_{i+1}$ if and only if $G_i$ can be drawn in $P_i$ with the bend point of $e_i$ arbitrarily close to $b$ inside $\mathcal{V}_i$. 

\begin{figure}[t]
	\centering
\begin{subfigure}[t]{.45\textwidth}
	\centering
	\includegraphics[page=2]{figs/independent-convex}
	\caption{}
	\label{fig:independent-convex}
\end{subfigure}
\hfill
\begin{subfigure}[t]{.45\textwidth}
	\centering
	\includegraphics[page=1]{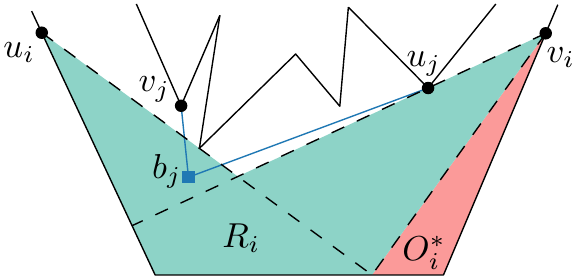}
	\caption{}
	\label{fig:independent-reflex}
\end{subfigure}
\caption{
(a)~$u_j$ and $v_j$ lie in $P(u_i,v_i)\circ \overline{v_iu_i}$,
but $R_i$ and $R_j$ are interior-disjoint.
(b)~$P(u_j,v_j)$ forces bend $b_j$ to lie inside $R_i$.
 Then, $b_j$ creates a reflex angle.}
\label{fig:interaction}
\end{figure}
 
For each convex edge $e_i=(u,v)$, we define the open \emph{restricted region} in $P_i$ as  $R_{i}=\mathrm{interior} \left[\left(\left(\bigcup_{b\in \mathcal{V}_i}O_{i}(b)\right)\setminus O^*_{i}\right)\cap \left(P_i(u,v)\circ \overline{vu}\right)\right]$; see Fig.~\ref{fig:independent-convex}.%
For each point $r\in R_{i}$, there are two convex bend points $b$ and $b'$ for $e_i$ such that bending $e_i$ at $b$ cuts off~$r$, whereas bending $e_i$ at $b'$ does not.
Note that $\partial R_i$ contains all minimal bend points for $e_i$.
Recall that we only process a convex edge $e_i$ if all siblings of $f_i$ in the dual tree $T_i$ are leaves and correspond to convex edges. We now prove that
the restricted region of $e_i$ does not interfere with the restricted regions
of these siblings.
 
\begin{lemma}\label{lem:convex-interaction}
Let $e_i,e_j$ be convex edges such that $f_i$ and $f_j$ are leaves and siblings in~$T_i$.
Then, the restricted regions $R_i$ and $R_j$ are interior-disjoint.
\end{lemma}

\begin{proof}
 Let $e_i =(u_i,v_i)$ and $e_j=(u_j,v_j)$.
 Since $f_i$ and $f_j$ are siblings, 
$P(u_i,v_i)$ and $P(u_j,v_j)$ do not share any interior point.

By definition, the restricted regions of $e_i$ and $e_j$ lie in the polygon $P(u_i,v_i)\circ \overline{v_iu_i}$ and the polygon $P(u_j,v_j)\circ\overline{v_ju_j}$, respectively; see Fig.~\ref{fig:independent-convex}.
Hence, if these two polygons are disjoint, the statement follows.
 W.l.o.g. assume that at least one of $u_j,v_j$ lies in $P(u_i,v_i)\circ \overline{v_iu_i}$.
This implies that $R_i(v_i,u_i)$
has exactly one bend that forms a reflex angle
in $R_i$. Further, by construction, there is no vertex in the interior of $R_i$;
in particular, neither $u_j$ nor $v_j$ can lie in the interior of $R_i$.
Therefore, no point of $P(u_j,v_j)\circ\overline{v_ju_j}$ lies in the interior of $R_i$.
\end{proof}

It follows from 
Lemma~\ref{lem:convex-interaction} that the bend points for $e_i$ and $e_j$
can be chosen independently.
Note that the only reason for $e_j$ to be drawn with its bend point inside $R_{i}$ is that some part of $P(u_j,v_j)$ intersects $\overline{u_jv_j}$; see Fig.~\ref{fig:independent-reflex}.
This would imply that $e_j$ is a reflex edge; this is why we process the leaves that correspond to reflex edges before the leaves that correspond to convex edges.

We are now ready to prove the correctness of our algorithm.
\begin{restatable}{lemma}{lemCorrectness}\label{lem:correctness}
	If $\mathrm{interior}(\mathcal{V}_i)=\emptyset$, then $G_i$ cannot be drawn in $P_i$.
	Otherwise, $G_i$ can be drawn in $P_i$ if and only if $G_{i+1}$ can be drawn in $P_{i+1}$.
\end{restatable}
\begin{proof}[Sketch]
	The first statement of the lemma follows immediately.
	
	For the second statement, first assume that we have a drawing $\Gamma_i$ 
	of~$G_i$ in~$P_i$.
	Recall that $P_{i+1}=P_i\setminus O_i^*$.  The bend point of $e_i$ cannot lie in $O_i^*$ and so $e_i$ is drawn in  $P_{i+1}$. 
	Hence, the edges $e_{i+1},\ldots,e_m$ also must be drawn in $P_{i+1}$.
	Thus, restricting $\Gamma_i$ to $G_{i+1}$, we obtain a drawing of $G_{i+1}$ in $P_{i+1}$.
		
  	Now assume that we have a drawing of~$G_{i+1}$ in~$P_{i+1}$
	where $e_i$ is drawn with bend point $b$.
	In the top-down traversal of the algorithm, we place the bend points of the
	edges corresponding	 to the children of $f_i$ in the interior of the polygon $P_i'=P_i(u,v)\circ\overline{vb}\circ\overline{bu}$.
	These edges cannot cross any edge $e_i,\ldots,e_m$.
	
	For any child $f_h$ of $f_i$, we can prove that there is still a valid
	bend point for $e_h$ inside $P_i'$. %
	For any pair of children $f_h,f_j$ of $f_i$, $e_h$ and $e_j$ do not cross. 
	This can be proven using Lemma~\ref{lem:reflex} or~\ref{lem:convex-interaction}, depending on whether  these edges are convex or reflex. %
	The full proof can be found in the appendix.
\end{proof}

We are now ready to prove Theorem \ref{thm:main}, which is the main result of this paper. We report the statement for the reader's convenience.

\MainTheorem*
\begin{proof}
The correctness follows immediately from Lemma~\ref{lem:correctness}.

For the running time,
we first argue that, for every $0\le i\le m$, $P_{i+1}$ has $O(n)$ corners.
If $e_i$ is a reflex edge, then obviously $P_{i+1}$ has fewer corners than $P_i$,
so assume that $e_i$ is a convex edge.
To create $P_{i+1}$ from $P_i$, we cut off the obstructed region $O_i^*$
of $e_i$. 
Through an easy charging argument, we can charge the new corners on
the boundary of $P_{i+1}$ to removed corners of $P_i$:
each of the new introduced edges on the boundary of $P_{i+1}$
has one end point that is a corner of $P_i$, and it cuts off at least
one corner from $P_i$.
Note that two newly introduced edges may form a new vertex, but both of these edges cut off at least one vertex.

In each step of the algorithm, we have to compute the visibility region $\mathcal{V}_i$ for~$e_i$.
Since $\mathcal{V}_i$ is a simple polygon with $O(n)$ edges, it can be computed in $O(n)$ time, as demonstrated by Gilbers~\cite[page 15]{Gilbers2014}.
Doing two traversals, the visibility region of each edge needs to be computed at most twice.
The remaining parts of the algorithm (computing the dual graph of $G$, choosing the order of the faces $f_i$ in which we traverse the graph, computing $O_i^*$, ``cutting off'' parts of $P_i$, and placing the bend points) can clearly be done within this time. 

Thus, the total running time is $O(nm)$.
\end{proof}

\section{Open Problems}
\label{sec:conclusion}

There are several interesting open problems related to our work.
What if we allow more than one bend per edge?
What if we allow some well-behaved crossings, e.g., outer-$1$-planar drawings~\cite{DBLP:journals/ijcga/DehkordiE12}?
The $\exists\mathbb{R}$-completeness proof by Lubiw, Miltzow, and Mondal~\cite{lubiw18} for any
graph with pre-drawn outer face allows overlaps between interior edges
and edges of the outer face. Is the problem still $\exists\mathbb{R}$-complete
if overlaps are forbidden?

 \bibliographystyle{splncs04}
\bibliography{abbrv,diss}

\newpage

\appendix
\section*{Appendix}
\lemCorrectness*
\begin{proof}
	The first statement of the lemma follows immediately since the bend point 
	of $e_i$ has to be placed in the interior of $\mathcal{V}_i$.
	
	For the second statement of the lemma, first assume that we have a drawing $\Gamma_i$ 
	of~$G_i$ in~$P_i$ where $e_i$ is drawn with bend point $b$.
	We have that $P_{i+1}=P_i\setminus O_i^*$. 
	Since $e_i$ cannot be drawn in $P_i$ such that its bend point lies in the
	obstructed region, $b$ lies in $P_{i+1}$. 
	Hence, in $\Gamma_i$ the edges $e_{i+1},\ldots,e_m$ also must be drawn in $P_{i+1}$.
	Thus, restricting $\Gamma_i$ to $G_{i+1}$, we obtain a drawing of $G_{i+1}$ in $P_{i+1}$.
		
  Now assume that we have a drawing of~$G_{i+1}$ in~$P_{i+1}$
	where $e_i$ is drawn with bend point $b$.
	In the top-down traversal of the algorithm, we place the bend points of the
	edges corresponding	to the children of $f_i$ in the interior of the polygon $P_i'=P_i(u,v)\circ\overline{vb}\circ\overline{bu}$.
	By construction, these edges cannot cross any edge $e_i,\ldots,e_m$.
	
	We first argue that for any child $f_h$ of $f_i$, there is still a valid
	bend point for $e_h$ inside $P_i'$.
	The interior of the visibility region $\mathcal{V}_h$ of $e_h$ was not empty in~$P_h$.
	The polygon $P_i'$ contains every point of $\mathcal{V}_h$ except those
	cut off by the drawing of $e_i$.
	Hence, we only have to
	argue that the bend point $b$ of $e_i$ cannot be placed such that the visibility 
	region of $e_h$ is empty in $P_i'$.
	
	If $e_h$ is a reflex edge, then the visibility region of $e_h$ in $P_i'$
	is only empty if $b$ lies
	either in the region obstructed by the minimal bend point of $e_h$, or on one of the segments between
	the end points of $e_h$ and its minimal bend point. However, by construction,
	these points do not lie in the interior of $P_{i+1}$, so they cannot be used
	as the bend point of $e_i$.
	
	If $e_h$ is a convex edge, then the restricted region of $e_h$ contains (at least partially) the visibility region of $e_h$.
	By definition, if $b$ lies in the restricted region of $e_h$, then there is 
	still a valid bend point for $e_h$. 
	If $b$ lies outside the restricted region, then the whole restricted region lies in $P_i'$.

	We still have to argue that for any two children $f_h,f_j$ of $f_i$ with $h<j$,
	$e_h$ and $e_j$ do not cross.

	If $e_h=(u,v)$ is a reflex edge, then by Lemma~\ref{lem:reflex} it
	has a unique minimal bend point $b$ and, by construction, the segments 
	$\overline{ub}$ and $\overline{bv}$ are part of $\partial P_j$.
	Hence, by drawing $e_h$ with its bend point close to $b$, 
	it cannot intersect $e_j$.
	
	If $e_h$ is a convex edge, then by construction $e_j$ is also a convex edge.
	Hence, by Lemma~\ref{lem:convex-interaction}, their restricted regions are
	interior-disjoint, so $e_h$ and $e_j$ cannot cross
	if they are drawn with a convex bend point.
\end{proof}

\end{document}